\documentclass[9pt,a4paper,twoside]{rho-class/rho}
\usepackage[english]{babel}

\makeatletter
\providecommand{\@license}{} 
\makeatother

\PassOptionsToPackage{dvipsnames}{xcolor}
\usepackage{amsmath,amssymb}
\usepackage{braket}
\usepackage{tikz}
\usepackage{xcolor}
\usepackage{amsthm}
\usepackage{enumitem}

\def\acceptedlanguage{English}
\def\publishedlanguage{English}

\definecolor{qpurple}{RGB}{112, 81, 145} 
\definecolor{qolive}{RGB}{110, 120, 85}
\definecolor{qbrick}{RGB}{150, 70, 65}

\theoremstyle{plain}
\newtheorem{theorem}{Theorem}[section]
\newtheorem{lemma}[theorem]{Lemma}

\theoremstyle{remark}

\theoremstyle{definition}

\usepackage{listings}
\usepackage{braket}

\definecolor{mathgreen}{RGB}{10,150,10} 
\lstdefinestyle{pseudocode}{
    language=Python,
    basicstyle=\ttfamily\small,
    mathescape=true,
    columns=fullflexible,
    keepspaces=true,
    upquote=true,
    literate=
      {*}{{$\bullet$}}1
      {•}{{$\bullet$}}1
      {$}{{{\color{mathgreen}$}}}1
      {^}{{\color{mathgreen}^}}1
      {_}{{\color{mathgreen}_}}1,
}

\setbool{rho-abstract}{true} 
\setbool{corres-info}{true}  

\title{Virtual Qudits for Simon’s Problem: Dimension Lifted Algorithms on Qubit Hardware}

\author[1,$\dagger$]{Abed Semre}
\author[2]{Steven Frankel}

\affil[1]{Computer Science Department, Technion – Israel Institute of Technology, Haifa, Israel}
\affil[2]{Faculty of Mechanical Engineering, Technion – Israel Institute of Technology, Haifa 3200003, Israel}

\dates{This manuscript was compiled on November 30, 2025}

\leadauthor{Semre et al.}
\smalltitle{Qudit Simon’s Algorithm}
\institution{Technion}
\theday{Nov 30, 2025} 

\email{abed.semre@campus.technion.ac.il}

\begin{abstract}
We study Simon’s problem over the module $\mathbb{Z}_d^n$ for arbitrary $d \ge 2$ and show how to explore qudit style advantages using only qubit based hardware and qubit level oracles. 
Starting from a standard binary ($d=2$) instance with promise oracle $U_f:\mathbb{Z}_2^n \to \mathbb{Z}_2^n$, we construct a dimension lifted oracle $f_{(d)}:\mathbb{Z}_d^n \to \mathbb{Z}_d^n$ for $d = 2^\ell$ by a simple layer wise pack/unpack encoding of $\ell$ qubits into one virtual qudit, and illustrating how qudit style formulations can be studied and exploited on standard qubit devices without access to native multilevel hardware.
For a hidden shift $s \in \mathbb{Z}_d^n$ of full order $\operatorname{ord}(s)=d$, one query to $f_{(d)}$ placed between two layers of $\mathrm{QFT}_d^{\otimes n}$ layer yields measurement outcomes that are exactly uniform on
\[
S^\perp \coloneqq \{\, y \in \mathbb{Z}_d^n : y^\top s \equiv 0 \pmod d \,\},
\]
recovering the original qubit algorithm when $d=2$.
From this structure we derive non-asymptotic bounds on the probability of obtaining $n-1$ independent constraints and obtain explicit repetition budgets as functions of the effective local dimension $d$, showing that the expected number of oracle calls remains $\Theta(n)$ while the required repetitions decrease with $d$.
Using QuTiP, we simulate these dimension lifted instances (for $d \in \{2,3,4\}$), confirming uniform sampling on $S^\perp$.
\end{abstract}


\begin{document}

\maketitle
\section{Background and Motivation}
\label{sec:background}

\subsection{From Qubits to Qudits}
\label{sec:qubits_to_qudits}

\rhostart{Q}uantum computation is most commonly described in terms of \emph{qubits}, i.e., two level systems with computational basis $\{\ket{0},\ket{1}\}$. 
At the hardware level, however, many platforms (e.g., trapped ions ~\cite{Bruzewicz_2019}, photonics ~\cite{Flamini_2018}) naturally exhibit more than two accessible energy levels per physical site. 
This motivates working with an explicitly multilevel \emph{qudit-based} based formalism, in which the local Hilbert space is $\mathcal{H}_d \cong \mathbb{C}^d$ with basis $\{\ket{0},\dots,\ket{d-1}\}$.

Qudits can increase information density~\cite{Wang_2020},
reduce the depth of certain circuits~\cite{Gokhale_2019},
and better exploit hardware native transitions~\cite{Bruzewicz_2019}.
These potential advantages have motivated qudit generalizations of
several canonical quantum primitives, including Fourier transforms
over $\mathbb{Z}_d$~\cite{NielsenChuang},
multivalued oracles~\cite{Muthukrishnan_2000},
and other interference based routines~\cite{Wang_2020} for dimensions $d>2$.

In this work, the term ``qudit'' is used in a slightly broader, algorithmic sense. 
We allow $d$ level systems to be either
(i) \emph{native} multilevel hardware, when available, or 
(ii) \emph{virtual} qudits obtained by encoding $\ell$ qubits into one effective $d = 2^\ell$ level system via a simple pack/unpack map.
All of our concrete constructions can be implemented using only qubit level operations and a standard binary Simon oracle; the qudit viewpoint is used to reveal how increasing the effective local dimension impacts the behavior of the algorithm.

\paragraph{Notation and assumptions.}
Unless stated otherwise we work over the ring $\mathbb{Z}_d$ with the standard inner product
\[
x \cdot y \;=\; \sum_i x_i y_i \pmod d.
\]
Vectors in $\mathbb{Z}_d^n$ label computational basis states $\ket{x} \in \mathcal{H}_d^{\otimes n}$, which may be realized either as native qudit registers or as encoded blocks of qubits.

\subsection{Why Generalize Simon's Algorithm?}
\label{sec:why_simon}

Simon’s algorithm~\cite{Simon_1994} is one of the earliest examples of an exponential separation between quantum and classical query complexity. 
Given a function $f:\{0,1\}^n \to \{0,1\}^n$ that is $2$-to-$1$ with promise $f(x) = f(x \oplus s)$ for some unknown $s \neq 0$, the quantum algorithm recovers $s$ using $O(n)$ oracle queries, whereas any classical randomized algorithm requires exponentially many queries in $n$ ~\cite{Simon_1994}.

While multilevel generalizations exist for other algorithms (such as Deutsch–Jozsa~\cite{DeutschJozsa_1992} and Grover search ~\cite{Grover_1997}), a careful qudit formulation of Simon’s algorithm is particularly attractive for two reasons:
(i) it makes explicit how higher dimensional interference structures control sampling uniformity over orthogonal subspaces and the number of repetitions needed to obtain $n-1$ independent constraints; and
(ii) it offers a clean setting in which to ask whether \emph{qudit style advantages} (e.g., fewer repetitions for a fixed failure probability) can be explored and quantified even when only qubit hardware and a binary oracle are available.

The present work uses Simon’s problem as a testbed: we first formulate the algorithm over $\mathbb{Z}_d^n$, and then show how such a $d$-ary instance can be \emph{implemented and simulated} on qubit based devices by encoding groups of qubits into effective qudits.

\subsection{Goals and Contributions of This Work}
\label{sec:goals}

At a high level, our goal is to highlight how increasing the (effective) local dimension $d$ affects the behavior of Simon’s algorithm, and to show that these qudit style effects can already be studied on qubit only platforms by a simple encoding construction. 
More concretely, this work provides:

\begin{itemize}
    \item A precise formulation of the $d$-to-one promise over $\mathbb{Z}_d^n$ with hidden shift $s \neq 0$ and the associated qudit version of Simon’s algorithm, including the characterization of the measurement outcomes as uniform samples from
    \[
    S^\perp \;=\; \{\, y \in \mathbb{Z}_d^n : y \cdot s \equiv 0 \pmod d \,\}.
    \]

    \item A \emph{qubit-native} construction of a $d$-to-one oracle $f_{(d)}:\mathbb{Z}_d^n \to \mathbb{Z}_d^n$ for $d = 2^\ell$ using only the original binary Simon oracle $U_f:\mathbb{Z}_2^n \to \mathbb{Z}_2^n$. 
    The construction is based on a layerwise pack/unpack encoding of $\ell$ qubits into one effective qudit, and preserves the hidden shift structure in a dimension lifted instance.

    \item A complexity analysis that bounds the probability of collecting $n-1$ independent samples from $S^\perp$ and derives explicit repetition counts as a function of the local dimension $d$. 
    The resulting bounds show that the expected number of oracle calls remains $\Theta(n)$, while the number of repetitions required to achieve a target failure probability decreases as $d$ grows.

    \item Numerical case studies and QuTiP simulations for representative dimensions (e.g., $d \in \{2,3,4\}$). 
    These simulations empirically confirm uniform sampling on $S^\perp$ and illustrate the predicted dimension repetition tradeoff.
\end{itemize}


\section{Quantum Gates and Operations for Qudit Based Algorithms}
\label{sec:gates_ops}

\subsection{States}
\label{sec:states}

A qudit lives in $\mathcal{H}_d = \mathrm{span}\{\,|i\rangle \mid i \in \mathbb{Z}_d\,\}$~\cite{NielsenChuang, Wang_2020} with orthonormal basis $\{|i\rangle\}_{i=0}^{d-1}$, $\langle i|j\rangle=\delta_{ij}$. Any pure state has the form
\begin{equation}
    |\psi\rangle = \sum_{i \in \mathbb{Z}_d} \alpha_i |i\rangle
    \quad \text{with} \quad
    \sum_{i \in \mathbb{Z}_d} |\alpha_i|^2 = 1,
    \label{eq:qudit_state}
\end{equation}
which we identify with a vector in $\mathbb{C}^d$.

\subsection{Gates}
\label{sec:gates}

\subsubsection{X gate}
\label{sec:xgate}

For qubits, $X$ flips $|0\rangle \leftrightarrow |1\rangle$. For qudits, the generalized shift acts by modular addition:
\begin{equation}
    X_d |i\rangle = |i \oplus 1\rangle,
    \qquad i \in \mathbb{Z}_d,
    \label{eq:xgate_def}
\end{equation}
with $\oplus$ modulo $d$. In matrix form, $X_d$ is the $d\times d$ cyclic permutation matrix:
\[
    X_d =
    \begin{pmatrix}
        0 & 0 & \dots & 0 & 1 \\
        1 & 0 & \dots & 0 & 0 \\
        0 & 1 & \dots & 0 & 0 \\
        \vdots & \vdots & \ddots & \vdots & \vdots \\
        0 & 0 & \dots & 1 & 0
    \end{pmatrix}.
\]
For $d=2$, \[
X_2=\begin{pmatrix}0 & 1\\ 1 & 0\end{pmatrix}.
\]

Such multi-valued increment operations are standard primitives in
qudit logic and appear in the construction of multi-valued
controlled gates and oracles~\cite{Muthukrishnan_2000}.

\subsubsection{H gate}
\label{sec:hgate}

The qudit analogue of the Hadamard is the QFT over $\mathbb{Z}_d$:
\begin{equation}
    \mathrm{QFT}_d |j\rangle = \frac{1}{\sqrt{d}} \sum_{k=0}^{d-1} \omega^{jk} |k\rangle,
    \quad \omega=e^{2\pi i/d},
    \label{eq:hgate_def}
\end{equation}
with matrix
\[
    \mathrm{QFT}_d = \frac{1}{\sqrt{d}}
    \begin{pmatrix}
        1 & 1 & 1 & \dots & 1 \\
        1 & \omega & \omega^{2} & \dots & \omega^{d-1} \\
        1 & \omega^{2} & \omega^{4} & \dots & \omega^{2(d-1)} \\
        \vdots & \vdots & \vdots & \ddots & \vdots \\
        1 & \omega^{d-1} & \omega^{2(d-1)} & \dots & \omega^{(d-1)^2}
    \end{pmatrix}.
\]
For $d=2$, \[
\mathrm{QFT}_2 = H=\frac{1}{\sqrt{2}}\begin{pmatrix}1&1\\ 1&-1\end{pmatrix}.
\]

\section{Generalized Simon’s Problem over Qudits}
\label{sec:simon_qudit}

Let \(f:\mathbb{Z}_d^n\to\mathbb{Z}_d^n\) be \(d\)-to-one with hidden shift \(s\neq 0\) such that
\[
    f(x)=f(y) \iff y=x\oplus k\,s \quad \text{for some } k\in\{0,\dots,d-1\},
\]
where \(\oplus\) denotes addition \((\bmod\, d)\).

\subsection{Complexity}
\label{sec:complexity_qudit}

\subsubsection{Classical complexity}
\label{sec:classical_qudit}
Any classical randomized algorithm requires exponentially many queries in \(n\). In particular, observing a collision \(f(x)=f(y)\) (which reveals a multiple of \(s\)) needs \(\Omega(d^{n/2})\) queries by a birthday paradox argument, and fully identifying \(s\) remains exponential in \(n\).

\subsubsection{Coset decomposition}
\label{sec:quantum_structure_qudit}
Partition \(\mathbb{Z}_d^n\) into \(d^{\,n-1}\) cosets (orbits) of \(\langle s\rangle\): choose representatives \(S_1=\{z_i^{(1)}\}\) and set \(z_i^{(j)}=z_i^{(1)}\oplus (j-1)s\), so that all \(z_i^{(j)}\) collide under \(f\).

\subsection{Quantum algorithm}
\label{sec:algo_qudit}

Given the oracle \(U_f:|x\rangle|0\rangle\mapsto|x\rangle|f(x)\rangle\), the algorithm is:
\begin{enumerate}
    \item \textbf{Create superposition.} Start with \(|\psi_0\rangle=|0\rangle^{\otimes n}|0\rangle^{\otimes n}\) and apply \(\mathrm{QFT}_d^{\otimes n}\) to the first register:
    \[
        |\psi_1\rangle=\frac{1}{d^{n/2}}\sum_{x\in\mathbb{Z}_d^n}|x\rangle|0\rangle^{\otimes n}.
    \]
    \item \textbf{Apply the oracle:}
    \[
        |\psi_2\rangle=\frac{1}{d^{n/2}}\sum_{x\in\mathbb{Z}_d^n}|x\rangle|f(x)\rangle.
    \]
    \item \textbf{Apply \(\mathrm{QFT}_d^{\otimes n}\) again:}
    \[
        |\psi_3\rangle=\frac{1}{d^{n}}\sum_{x,y\in\mathbb{Z}_d^n}\omega^{\,x\cdot y}\,|y\rangle\,|f(x)\rangle,
        \qquad \omega=e^{2\pi i/d}.
    \]
    \item \textbf{Exploit the promise.} Group by orbits with representatives \(z\in S_1\) and write \(x=z\oplus k s\):
    \begin{align}
        |\psi_3\rangle
        &= \frac{1}{d^{n}} \sum_{y} |y\rangle
           \left[\sum_{z\in S_1} \omega^{\,z\cdot y}
                 \left(\sum_{k=0}^{d-1} (\omega^{\,s\cdot y})^k \right)
                 |f(z)\rangle \right].
        \label{eq:psi3_expanded}
    \end{align}
    The inner geometric sum equals \(d\) if \(s\cdot y\equiv 0\ (\bmod\, d)\) and \(0\) otherwise; thus only \(y\in S^\perp:=\{y\in\mathbb{Z}_d^n:\, y\cdot s\equiv 0\ (\bmod\, d)\}\) survive, giving
    \[
        |\psi_3\rangle=\frac{1}{d^{n-1}} \sum_{y\in S^\perp} |y\rangle
        \left[\sum_{z\in S_1} \omega^{\,z\cdot y}|f(z)\rangle\right].
    \]
    \item \textbf{Measurement.} Measuring the first register yields \(y\in S^\perp\) and a linear constraint \(y\cdot s\equiv 0\ (\bmod\, d)\). Repeating yields \(n-1\) independent equations in expected \(\Theta(n)\) samples.
\end{enumerate}

\paragraph{Remark.}
Setting \(d=2\) recovers Simon’s original qubit algorithm.

\paragraph{Probability distribution over \(S^{\perp}\).}
For any fixed \(y'\in S^\perp\),
\[
    P(y') = |\langle y'|\psi_3\rangle|^2
           = \left(\frac{1}{d^{n-1}}\right)^2
              \sum_{z',z\in S_1}(\omega^{\,z'\cdot y'})^* \omega^{\,z\cdot y'} \,\delta_{f(z),f(z')}
           = \frac{1}{d^{n-1}},
\]
confirming uniformity over \(S^\perp\).


\section{Complexity and Dimension Dependence}
\label{sec:effect_of_d}

Let \(|S^\perp|=d^{\,n-1}\). After \(m\) independent samples, the span contains at most \(d^{\,m}\) elements; hence the probability that the next sample is independent satisfies
\[
    \Pr[\text{independent at step }m+1] \;\ge\; 1-\frac{d^{\,m}}{d^{\,n-1}}.
\]
Therefore
\[
    p \;\ge\; \prod_{i=0}^{n-2} \left(1-\frac{d^{\,i}}{d^{\,n-1}}\right),
\]
which is exact when \(d\) is prime (vector space case), and a valid lower bound otherwise. Consequently, for one full run
\[
    P_{\mathrm{fail}}^{(1)} \;\le\; 1-p \;\le\; \frac{d+1}{d^{2}} - \frac{1}{d^{\,n}},
\]
and after \(k\) independent runs,
\[
    P_{\mathrm{fail}}^{(k)} \;\le\; \left(\frac{d+1}{d^{2}} - \frac{1}{d^{\,n}}\right)^{k}.
\]
To ensure \(P_{\mathrm{fail}}^{(k)} \le \epsilon\),
\[
    k \;\ge\; \left\lceil \log_{\left(\frac{d+1}{d^{2}} - \frac{1}{d^{\,n}}\right)} \!\!\left(\epsilon\right) \right\rceil
    \;\approx\; \left\lceil \frac{-\log \epsilon}{\,2\log d - \log(d+1)\,} \right\rceil
    \quad (\text{large }n).
\]

\begin{figure}[h]
    \centering
    \includegraphics[width=0.8\columnwidth]{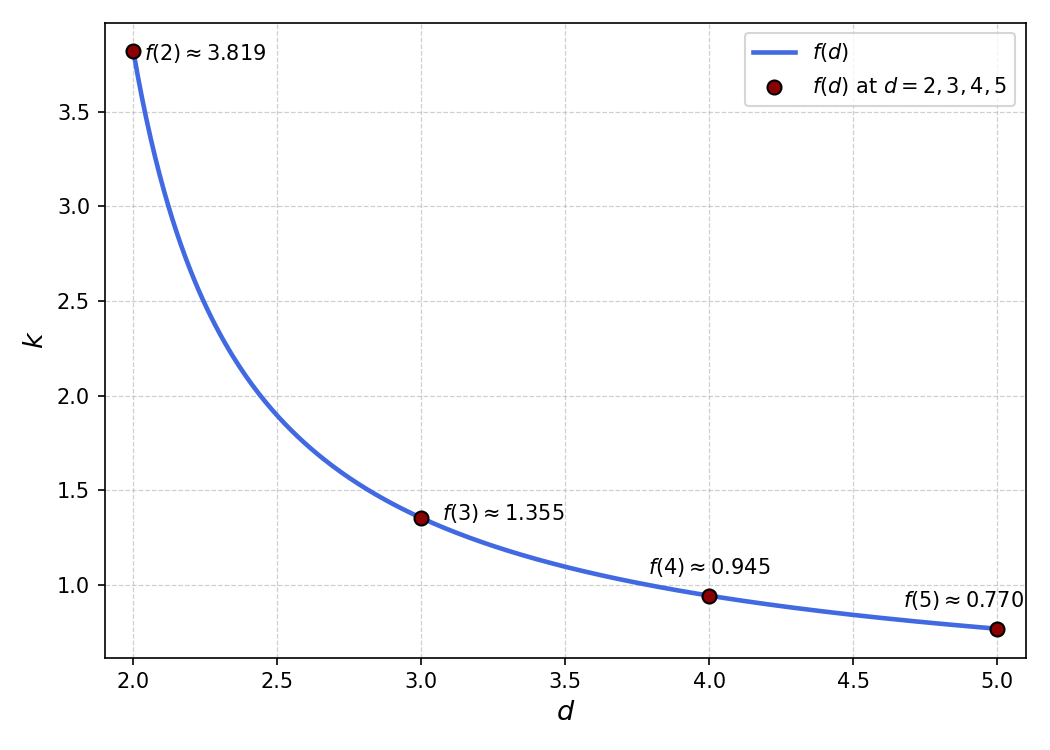}
    \caption{Asymptotic repetitions \(k\) required for \(P_{\mathrm{fail}}\le 1/3\) as a function of \(d\):
    \(f(d)=\frac{\log 3}{2\log d - \log(d+1)}\).}
    \label{fig:k_vs_d}
\end{figure}

With target failure $\epsilon\in(0,1)$ and large $n$,
the base decreases with $d$ for $d\ge 2$, so $k$ is monotonically decreasing in $d$. A single shot condition $k=1$ requires, asymptotically,
\[
    \epsilon \geq \frac{d+1}{d^2}.
\]
Illustrative computations for $\epsilon\in\{10^{-1},10^{-2},10^{-3}\}$ 
reveal how the admissible region of $(d,n)$ values shrinks as the error tolerance tightens. 
Figure \ref{fig:surface_epsilons} summarizes the behavior of 
$f(d,n,\epsilon)=\log(\epsilon)/\log\!\big((d+1)/d^2-d^{-n}\big)$ 
across representative $\epsilon$ levels.

\begin{figure}[h!]
    \centering
    \begin{subfigure}[b]{0.32\textwidth}
        \includegraphics[width=\textwidth]{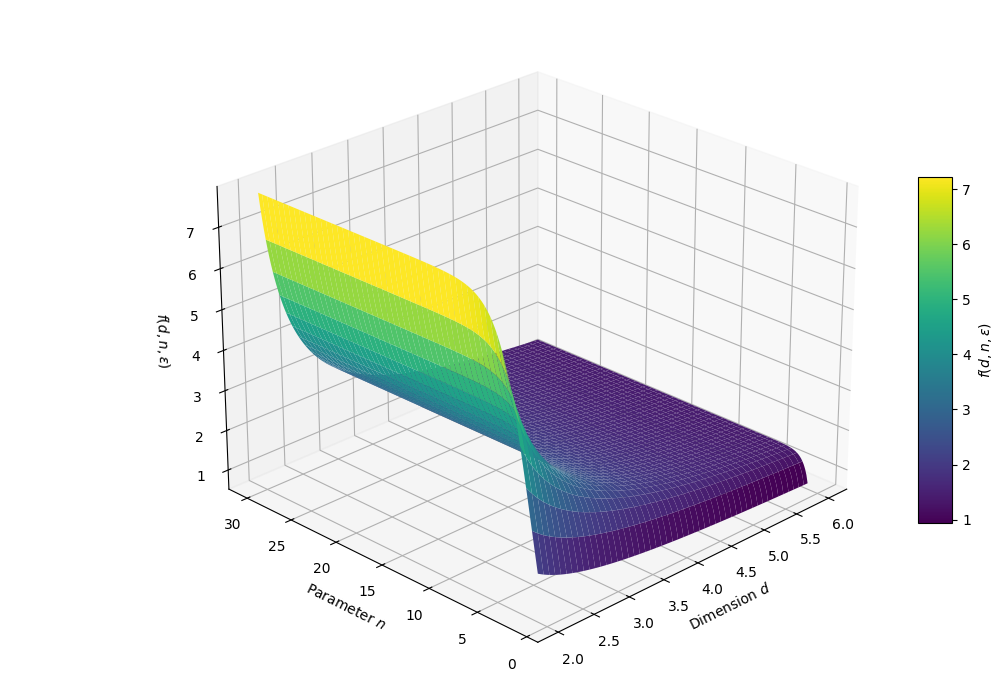}
        \caption{$\epsilon=10^{-1}$}
    \end{subfigure}\hfill
    \begin{subfigure}[b]{0.32\textwidth}
        \includegraphics[width=\textwidth]{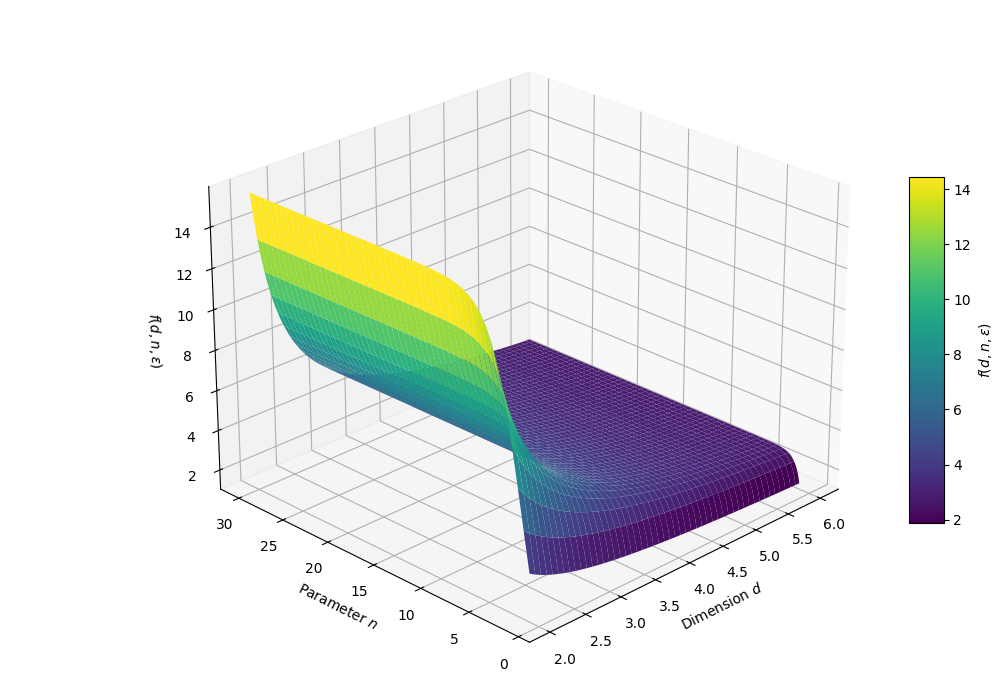}
        \caption{$\epsilon=10^{-2}$}
    \end{subfigure}\hfill
    \begin{subfigure}[b]{0.32\textwidth}
        \includegraphics[width=\textwidth]{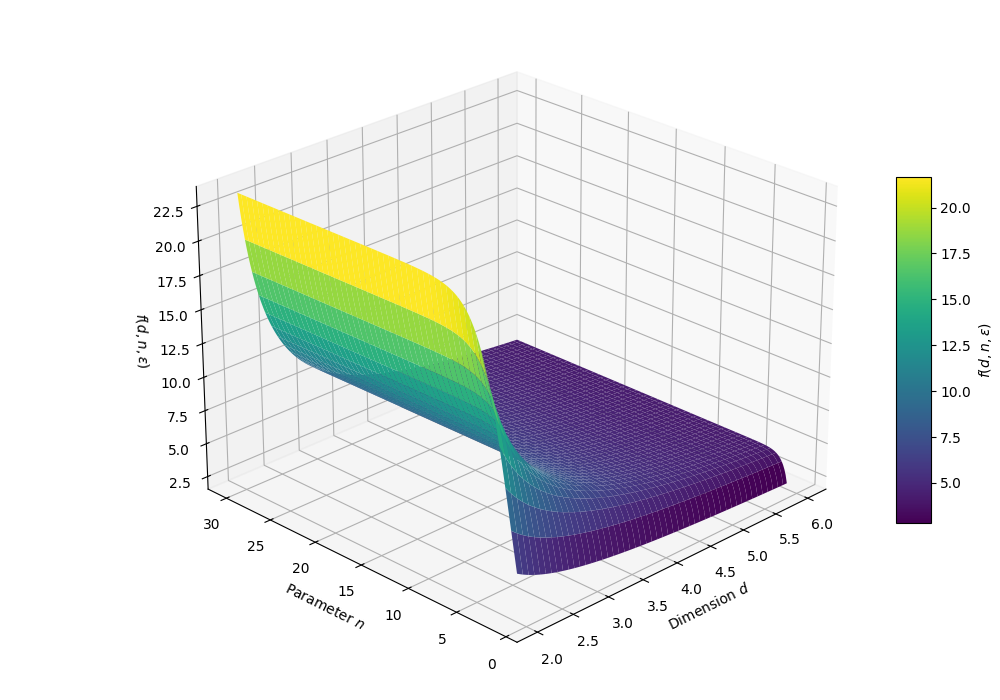}
        \caption{$\epsilon=10^{-3}$}
    \end{subfigure}
    \caption{Effect of decreasing $\epsilon$ on the surface $f(d,n,\epsilon)$. Lower values of $\epsilon$ increase the required repetition count, while larger dimensions $d$ reduce it, shifting the feasible region toward smaller $k$ as $d$ grows.}
    \label{fig:surface_epsilons}
\end{figure}

\paragraph{Illustrative thresholds.}
Table~\ref{tab:dimension_repetitions} summarizes the tradeoff between 
dimension and repetitions for several target failure probabilities. 
As $\epsilon$ decreases, the required number of repetitions for qubits 
grows only logarithmically, while the single shot threshold dimension 
$d_{\text{single-shot}}$ increases by roughly one order of magnitude 
for each additional digit of precision.

\begin{table}[h]
    \centering
    \caption{Approximate single shot thresholds and repetition counts
    for Simon's algorithm at different error tolerances $\epsilon$.
    The column $d_{\text{single-shot}}$ gives the smallest dimension
    (rounded) for which the asymptotic bound yields $k(d) \approx 1$,
    while $k(d=2)$ is the corresponding number of repetitions for qubits.}
    \label{tab:dimension_repetitions}
    \begin{tabular}{@{}ccc@{}}
        \toprule
        Target failure $\epsilon$ 
        & $d_{\text{single-shot}}$ (i.e., $k(d)\approx 1$) 
        & Repetitions $k(d=2)$ \\
        \midrule
        $10^{-1}$ & $\approx 11$   & $\approx 8$  \\
        $10^{-2}$ & $\approx 101$  & $\approx 16$ \\
        $10^{-3}$ & $\approx 1001$ & $\approx 25$ \\
        \bottomrule
    \end{tabular}
\end{table}


\subsection{Dimensional Lifting via Function Multiplicity Expansion}
\label{sec:dim_lifting}

Beyond the hardware’s native dimension $d$, one may \emph{embed} a $d$-to-one promise function into a higher dimensional $(d'\! =\! \ell d)$-to-one instance, thereby tightening the single shot failure bound 
$P_{\mathrm{fail}}\le (d+1)/d^2$.  
The resulting bound satisfies
\[
  \frac{d'+1}{(d')^2} = \frac{\ell d + 1}{(\ell d)^2}.
\]
Selecting the smallest integer $\ell$ such that
\[
  \frac{\ell d + 1}{(\ell d)^2} \leq \epsilon
\]
guarantees one-shot success probability exceeding $1-\epsilon$.  
Rearranging yields the quadratic condition
\[
  \epsilon d^2 \ell^2 - d\ell - 1 \geq 0,
  \qquad
  \Rightarrow\qquad
  \ell \geq \frac{1 + \sqrt{1 + 4\epsilon}}{2\epsilon d}.
\]

This expression quantifies how much the effective dimension must be increased to achieve a given error tolerance.  
For instance, with $d=6$ and $\epsilon=10^{-2}$, the bound gives 
$\ell > 16.83$, so $\ell=17$ suffices, corresponding to $d' = 102$ and 
$P_{\mathrm{fail}}\approx 103/102^2 \approx 9.9\times10^{-3}$.


\section{Simulation of the Qudit Simon Algorithm}
\label{sec:simulation}

\noindent
We validate the measurement statistics predicted by the qudit variant of Simon’s algorithm by simulating the full unitary circuit in QuTiP.
Given local dimension $d$, register size $n$, and hidden shift $s\in\mathbb{Z}_d^n$, the simulation performs:
(i) state preparation on two $n$ qudit registers,
(ii) a generalized Hadamard (QFT$_d$) on the first register,
(iii) a black box call to $U_f$ for a $d$-to-one promise function with orbits $\{x+k s\}_{k\in\mathbb{Z}_d}$,
(iv) projective measurement of the second register to a value in $\mathrm{Im}(f)$, and
(v) a second QFT$_d$ on the first register followed by measurement.

\bigskip
\noindent\textbf{Minimal pseudocode.}
\vspace{-0.5em}
\begin{lstlisting}[style=pseudocode,
caption={Simulation pipeline (pseudocode).}, label={lst:sim_pseudo}]
• Inputs:
•   d : local dimension
•   n : number of qudits per register
•   s : hidden shift in $\mathbb{Z}_d^n$
•   M : number of trials

• $QFT_d$ = $d\times d$ QFT matrix
• $QFT_{\text{first}} = QFT_d^{\otimes n} \otimes I^{\otimes n}$ (QFT on first register only)
• $U_f$ : $|x\rangle|0\rangle \mapsto |x\rangle|f(x)\rangle$  (d-to-one structure)

• $|\psi_0\rangle = |0^{\otimes 2n}\rangle$
• $|\psi_1\rangle = QFT_{\text{first}} \, |\psi_0\rangle$
• $|\psi_2\rangle = U_f \, |\psi_1\rangle$

• Measure second register $\to$ collapse to $f(x_0)$
• $|\psi_3\rangle = QFT_{\text{first}}$  (collapsed state)

• Repeat M times:
• sample $y$ from measuring the first register

• Return:
•   samples $\{y\}$
•   support fraction on $S^\perp$ where $\langle y, s\rangle \equiv 0 \pmod d$
\end{lstlisting}

\subsection{Example: Simon's Algorithm with $d=4, n=4$}
\label{sec:example_d4}

Consider the hidden shift $s = [2,0,3,1] \in \mathbb{Z}_4^4$.  
Executing the generalized Simon algorithm repeatedly produces, with high probability, three independent measurement outcomes $y \in S^{\perp}$, sufficient to reconstruct $s$ via linear solving over $\mathbb{Z}_4$. The resulting empirical distribution is observed to be uniform over the orthogonal subspace of size $|S^{\perp}| = 4^3 = 64$.

\begin{figure}[h]
    \centering
    \includegraphics[width=0.9\columnwidth]{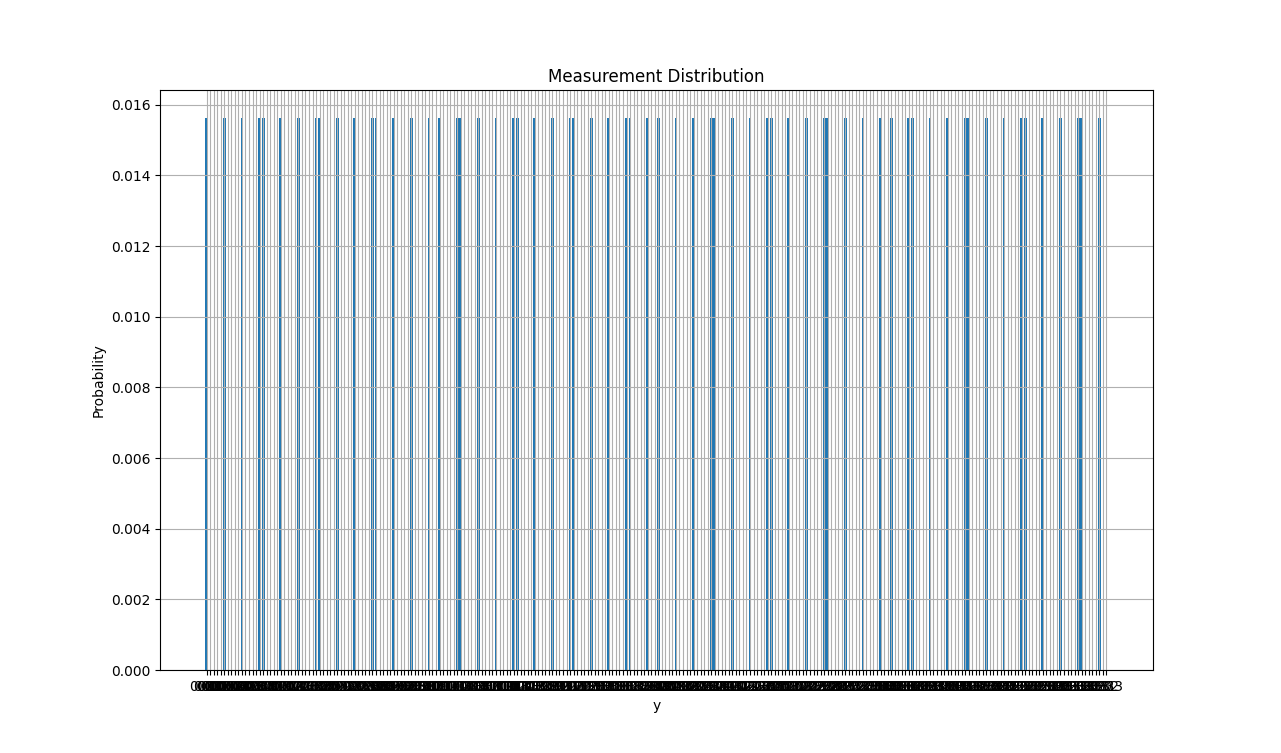}
    \caption{Empirical measurement distribution for $d=4$, $n=4$: samples are uniformly distributed over the orthogonal subspace $S^{\perp}$ containing $4^3 = 64$ outcomes.}
    \label{fig:simon_d4}
\end{figure}

\section{Lifting Simon’s Algorithm from Qubits to Qudits}
\label{sec:from_qubits_to_qudits}

Quantum algorithms formulated for qubits can often be generalized to higher dimensional systems \emph{qudits} to exploit their richer state space and potential for increased computational efficiency. A qudit of dimension \( d = 2^{\ell} \) can encode the information of \( \ell \) qubits, effectively compressing multiple binary subsystems into a single multilevel quantum unit. This mapping allows an exponential expansion of the accessible Hilbert space while maintaining only a linear increase in the number of physical qubits.

\subsection{Conceptual Framework}

The idea is to reinterpret groups of \(\ell\) qubits as one qudit, thus replacing each binary subsystem with a \(d\)-level subsystem. The same oracle \(U_f\) from the qubit version of Simon’s problem can then be applied on these qudit registers without modification, as the logical structure of the function \(f\) and its hidden symmetry remain consistent under modular arithmetic.

Formally, the standard initialization state of the Simon algorithm is:
\[
\ket{0}^{\otimes n} \ket{0}^{\otimes n}_{\text{ancilla}}.
\]
To lift this construction to qudits of dimension \( d = 2^{\ell} \), we append \(\ell - 1\) additional qubits per logical qubit, producing
\[
\ket{0_\ell}^{\otimes n} \ket{0_\ell}^{\otimes n}_{\text{ancilla}},
\]
where each \(\ket{0_\ell}\) denotes a logical zero state of one qudit composed of \(\ell\) physical qubits.

\subsection{Qubit\texorpdfstring{$\to$}{->}Qudit Lift Using Only the Binary Oracle}

\paragraph{Notation.}
Fix an integer $d=2^{\ell}$.
We write $\mathbb{Z}_2^n$ with bitwise XOR $\oplus$ and $\mathbb{Z}_d^n$ with addition modulo $d$, denoted $\boxplus$.
For $u\in\mathbb{Z}_d$, let its binary expansion be
$u=\sum_{t=0}^{\ell-1}2^t\,u^{(t)}$ with $u^{(t)}\in\{0,1\}$.
For $\eta=(\eta_0,\dots,\eta_{n-1})\in\mathbb{Z}_d^n$ define the \emph{binary layers}
\[
X^{(t)}(\eta)\;\coloneqq\; \big(\eta_0^{(t)},\eta_1^{(t)},\dots,\eta_{n-1}^{(t)}\big)\in\mathbb{Z}_2^n
\quad (t=0,\dots,\ell-1),
\]
and the \emph{packing} map $\mathrm{pack}_\ell:\big(\mathbb{Z}_2^n\big)^{\ell}\to\mathbb{Z}_d^n$ by
\[
\mathrm{pack}_\ell\big(Y^{(0)},\dots,Y^{(\ell-1)}\big)
\;\coloneqq\;
\sum_{t=0}^{\ell-1} 2^t\,Y^{(t)}\;\;(\text{componentwise in }\mathbb{Z}_d).
\]
Given $s\in\mathbb{Z}_2^n\setminus\{0\}$, set its \emph{layer replicated lift}
\[
s_{(\ell)} \;\coloneqq\;
\big( \underbrace{s}_{t=0},\underbrace{s}_{t=1},\dots,\underbrace{s}_{t=\ell-1} \big)\in\big(\mathbb{Z}_2^n\big)^{\ell}.
\]
We will use the shorthand
\[
\eta \;\oplus_{\text{layers}}\; (\delta_0,\dots,\delta_{\ell-1})\!\cdot\! s
\;\coloneqq\;
\mathrm{pack}_\ell\!\Big(X^{(0)}(\eta)\oplus \delta_0 s,\;\dots,\;X^{(\ell-1)}(\eta)\oplus \delta_{\ell-1} s\Big),
\]
for any $(\delta_0,\dots,\delta_{\ell-1})\in\{0,1\}^{\ell}$.

\paragraph{Lifted oracle (constructed from $U_f$ only).}
Let $U_f$ implement a $2$-to-$1$ promise function
$f:\mathbb{Z}_2^n\to\mathbb{Z}_2^n$ with $f(x)=f(x\oplus s)$ for some nonzero $s\in\mathbb{Z}_2^n$.
Define $f_{(d)}:\mathbb{Z}_d^n\to\mathbb{Z}_d^n$ by
\begin{equation}
\label{eq:lifted-oracle}
f_{(d)}(\eta)
\;\coloneqq\;
\mathrm{pack}_\ell\!\big(
f(X^{(0)}(\eta)),\;f(X^{(1)}(\eta)),\;\dots,\;f(X^{(\ell-1)}(\eta))
\big).
\end{equation}
Operationally: \emph{unpack} $\eta$ into its $\ell$ binary layers, apply the original $U_f$ to each layer, then \emph{pack} the $\ell$ binary outputs back into one $d$-ary vector.

\begin{lemma}[Oracle invariance for the lifted construction]
\label{lem:layer-invariance}
For all $\eta\in\mathbb{Z}_d^n$ and all $(\delta_0,\dots,\delta_{\ell-1})\in\{0,1\}^{\ell}$,
\[
f_{(d)}(\eta)
\;=\;
f_{(d)}\!\big(\eta \;\oplus_{\mathrm{layers}}\; (\delta_0,\dots,\delta_{\ell-1})\!\cdot\! s\big).
\]
Consequently, $f_{(d)}$ is $d$-to-$1$, with each fiber equal to the orbit
$\{\eta \;\oplus_{\mathrm{layers}}\; \delta\!\cdot\! s:\; \delta\in\{0,1\}^{\ell}\}$.
\end{lemma}

\begin{proof}
By Simon's promise, $f(x)=f(x\oplus s)$ for all $x\in\mathbb{Z}_2^n$.
Fix any $\eta$ and any $\delta\in\{0,1\}^{\ell}$.
For each layer $t$ we have
$f\!\big(X^{(t)}(\eta)\oplus \delta_t s\big)=f\!\big(X^{(t)}(\eta)\big)$.
Packing the $\ell$ equalities with \eqref{eq:lifted-oracle} yields the claim.
The last statement follows because $\{0,1\}^{\ell}$ has size $2^{\ell}=d$.
\end{proof}

After applying $\mathrm{QFT}_{2^\ell}^{\otimes n}$ and measuring, convert the outcome to its modulo-$d$ form in $\mathbb{Z}_{2^\ell}^n$ by reading each block of $\ell$ wires as a single base-$d$ digit.

\begin{figure}[t]
\centering
\begin{tikzpicture}[x=0.9cm, y=0.5cm, thick]

  \draw[qpurple, very thick] (0,  0) node[left] {$\ket{\psi}^{(0)}_{0}$}
        -- (0.5,  0) -- (2.1, 0) -- (2.5,0);
  \draw[qpurple, very thick] (3.5,0) -- (3.6,0) -- (6,0);

  \draw[qolive, very thick] (0, -1) node[left] {$\ket{\psi}^{(1)}_{0}$}
        -- (0.5, -1) -- (2.1,-3) -- (2.5,-3);
  \draw[qolive, very thick] (3.5,-3) -- (3.9,-3) -- (5.5,-1) -- (6,-1);

\node at (-0.5,-2) {$\vdots$};
\node at (6,-2) {$\vdots$};

  \draw[qbrick, very thick] (0, -3) node[left] {$\ket{\psi}^{(\ell-1)}_{0}$}
        -- (0.5, -3) -- (2.1,-8) -- (2.5,-8);
  \draw[qbrick, very thick] (6,-3) -- (5.5,-3) -- (3.9,-8) -- (3.5,-8);

  \draw[qpurple, very thick] (0, -4) node[left] {$\ket{\psi}^{(0)}_{1}$}
        -- (0.5, -4) -- (2.1,-2/3) -- (2.5,-2/3); 
  \draw[qpurple, very thick] (6,-4) -- (5.5,-4) -- (3.9,-2/3) -- (3.5,-2/3);

  \draw[qolive, very thick] (0, -5) node[left] {$\ket{\psi}^{(1)}_{1}$}
        -- (0.5, -5) -- (2.1,-11/3) -- (2.5,-11/3);
  \draw[qolive, very thick] (6,-5) -- (5.5,-5) -- (3.9,-11/3) -- (3.5,-11/3);

\node at (-0.5,-6) {$\vdots$};
\node at (6,-6) {$\vdots$};

  \draw[qbrick, very thick] (0, -7) node[left] {$\ket{\psi}^{(\ell-1)}_{1}$}
        -- (0.5, -7) -- (2.1,-26/3) -- (2.5,-26/3);
  \draw[qbrick, very thick] (6,-7) -- (5.5,-7) -- (3.9,-26/3) -- (3.5,-26/3);

  \draw[qpurple, very thick] (0, -8) node[left] {$\ket{\psi}^{(0)}_{n}$}
        -- (0.5, -8) -- (2.1,-2) -- (2.5,-2); 
  \draw[qpurple, very thick] (6,-8) -- (5.5,-8) -- (3.9,-2) -- (3.5,-2);

  \draw[qolive, very thick] (0, -9) node[left] {$\ket{\psi}^{(1)}_{n}$}
        -- (0.5, -9) -- (2.1,-5) -- (2.5,-5);
  \draw[qolive, very thick] (6,-9) -- (5.5,-9) -- (3.9,-5) -- (3.5,-5);

\node at (-0.5,-10) {$\vdots$};
\node at (6,-10) {$\vdots$};

  \draw[qbrick, very thick] (0, -11) node[left] {$\ket{\psi}^{(\ell-1)}_{n}$}
        -- (0.5, -11) -- (2.1,-10) -- (2.5,-10); 
  \draw[qbrick, very thick] (6,-11) -- (5.5,-11) -- (3.9,-10) -- (3.5,-10);

\filldraw[
  draw=qpurple,
  fill=qpurple!20,
  line width=1pt
] (2.5, -2) rectangle (3.5, 0);
\node at (3, -1) {$U_f$};

\node at (2.4,-4/3) {$\vdots$};
\node at (3.6,-4/3) {$\vdots$};

\filldraw[
  draw=qolive,
  fill=qolive!20,
  line width=1pt
] (2.5, -5) rectangle (3.5, -3);
\node at (3, -4) {$U_f$};

\node at (2.4,-13/3) {$\vdots$};
\node at (3.6,-13/3) {$\vdots$};

\node at (3,-6.5) {$\vdots$};

\filldraw[
  draw=qbrick,
  fill=qbrick!20,
  line width=1pt
] (2.5, -10) rectangle (3.5, -8);
\node at (3, -9) {$U_f$};

\node at (2.4,-28/3) {$\vdots$};
\node at (3.6,-28/3) {$\vdots$};

\end{tikzpicture}

\caption{
Schematic representation of the lifted oracle construction.  
Each horizontal wire corresponds to a physical qubit $\ket{\psi}^{(k)}_{j}$, where  
$k \in \{0,\ldots,\ell-1\}$ indexes the layer within a logical qudit and  
$j \in \{0,\ldots,n-1\}$ indexes the logical qudit position.  
For each fixed $j$, the $\ell$ wires detour through a common oracle block $U_f$,  
indicating that all $\ell$ physical qubits jointly encode a single $d = 2^\ell$-dimensional qudit.}
\end{figure}
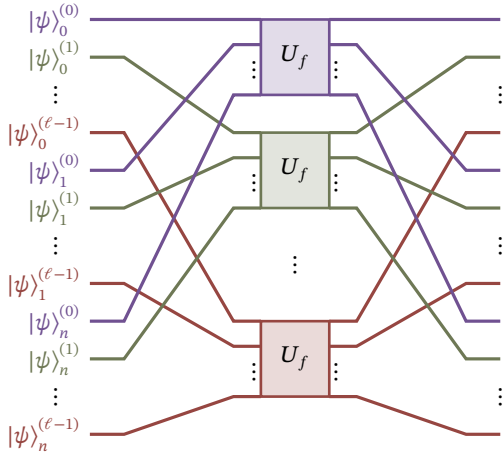

\subsection{Illustrative Construction of the Oracle and Algorithm over Qudits}
\label{sec:example_qudit}

This section gives a concrete worked example of the lifted oracle $f_{(d)}$ and the subsequent Fourier sampling when Simon’s algorithm is moved from qubits to qudits via the layerwise (pack/unpack) construction.

\subsubsection{Example: Oracle construction for $d=4$ (two layers)}
Let $d=4$ (so $\ell=2$) and take the binary hidden string
\[
s = 0101 \in \mathbb{Z}_2^{4}.
\]
Write any $\eta\in\mathbb{Z}_4^{4}$ in binary layers as
\[
X^{(0)}(\eta),\,X^{(1)}(\eta)\in\mathbb{Z}_2^{4},
\qquad
\eta=\mathrm{pack}_2\big(X^{(0)}(\eta),X^{(1)}(\eta)\big).
\]
Recall the lifted oracle (cf.\ Eq.~\eqref{eq:lifted-oracle})
\[
f_{(4)}(\eta)=
\mathrm{pack}_2\!\Big(
f\big(X^{(0)}(\eta)\big),\;
f\big(X^{(1)}(\eta)\big)
\Big),
\]
and the layerwise invariance (Lemma~\ref{lem:layer-invariance})
\[
f_{(4)}(\eta)
=
f_{(4)}\!\Big(
\mathrm{pack}_2\big(X^{(0)}(\eta)\oplus \delta_0 s,\;
                     X^{(1)}(\eta)\oplus \delta_1 s\big)
\Big),
\quad
(\delta_0,\delta_1)\in\{0,1\}^2.
\]
Thus every input $\eta$ has an orbit of size $4$ under the two independent layer toggles, and $f_{(4)}$ is constant on this orbit.

\paragraph{Concrete coordinate view.}
Let $\eta=(x_1,x_2,x_3,x_4)\in\mathbb{Z}_4^{4}$ and write each coordinate as
$x_j=2\,a_j+b_j$ with $a_j,b_j\in\{0,1\}$ (so $a_j=X^{(1)}_j(\eta)$ and $b_j=X^{(0)}_j(\eta)$).
Since $s=0101$, the indices with $s_j=1$ are $\{2,4\}$ (1-based).
Toggling layer $t$ adds $2^t$ to those coordinates modulo $4$.
Hence, for $(\delta_0,\delta_1)\in\{0,1\}^2$,
\[
\eta^{(\delta_0,\delta_1)}
\;=\;
\bigl(
x_1,\;\;x_2 \boxplus k,\;\;x_3,\;\;x_4 \boxplus k
\bigr),
\quad
k \;=\; \delta_0\cdot 1 \;+\; \delta_1\cdot 2 \;\in\;\{0,1,2,3\},
\]
where $\boxplus$ is addition modulo $4$ on $\mathbb{Z}_4$.
Equivalently, the four inputs in the fiber through $\eta$ can be listed succinctly as
\[
(x_1,\;x_2\boxplus k,\;x_3,\;x_4\boxplus k),
\qquad
k\in\{0,1,2,3\}.
\]
By Lemma~\ref{lem:layer-invariance},
\[
f_{(4)}(x_1,\,x_2,\,x_3,\,x_4)
=
f_{(4)}(x_1,\,x_2\boxplus k,\,x_3,\,x_4\boxplus k)
\quad\text{for all }k\in\{0,1,2,3\}.
\]
This clarifies (and replaces) the informal pattern ``$(x,k,y,k)$'': the same offset $k$ is added modulo $4$ exactly at the coordinates where $s_j=1$, while the others remain free.

Applying $\mathrm{QFT}_4^{\otimes 4}$ on each binary layer yields outcomes
\[
Y=\mathrm{pack}_2\big(Y^{(0)},Y^{(1)}\big),
\qquad
Y^{(0)},Y^{(1)}\in\mathbb{Z}_2^{4},
\]
with the standard Simon constraints
\[
\langle Y^{(t)}, s\rangle \equiv 0 \pmod 2,
\qquad t=0,1.
\]

\subsubsection{Example: Full Algorithmic Flow over Qudits}

Let us now explicitly trace the algorithmic steps for \(n = 2\) and hidden string \(s = 01\), promoted to \(d = 4\) (i.e., \(\ell = 2\)). Each logical qubit becomes a ququart.

The initial state is
\[
\ket{0}^{\otimes 2}\ket{0}^{\otimes 2}_{\text{ancilla}} = \ket{00,00}.
\]

\paragraph{Step 1: Apply the QFT.}
Applying $\mathrm{QFT}_4^{\otimes 4}$ to the ququart register yields outcomes:
\[
\frac{1}{4}
\Big(
\ket{0_4}+\ket{1_4}+\ket{2_4}+\ket{3_4}
\Big)
\otimes
\Big(
\ket{0_4}+\ket{1_4}+\ket{2_4}+\ket{3_4}
\Big)
\ket{0_4,0_4},
\]
or equivalently,
\[
\frac{1}{4}
\sum_{a,b=0}^{3} \ket{a,b}\ket{00}.
\]

\paragraph{Step 2: Apply the oracle.}
The oracle \(U_f\) acts as
\[
\ket{x}\ket{0} \longmapsto \ket{x}\ket{f(x)}.
\]
Thus, after \(U_f\),

\begin{align*}
\frac{1}{4}\!
\Big(
[\ket{00}+\ket{01}+\ket{02}+\ket{03}]\ket{f(00)} +\\
[\ket{10}+\ket{11}+\ket{12}+\ket{13}]\ket{f(10)} +\\
[\ket{20}+\ket{21}+\ket{22}+\ket{23}]\ket{f(20)} +\\
[\ket{30}+\ket{31}+\ket{32}+\ket{33}]\ket{f(30)}
\Big).
\end{align*}

\paragraph{Step 3: Measurement and collapse.}
Suppose the ancilla is measured in the state \(\ket{f(30)}\).  
The system collapses to
\[
\frac{1}{2}\big(\ket{30}+\ket{31}+\ket{32}+\ket{33}\big)\ket{f(30)},
\]
or, neglecting the ancilla,
\[
\frac{1}{2}\ket{3}\big(\ket{0}+\ket{1}+\ket{2}+\ket{3}\big).
\]
In the binary notation of two qubits per ququart,
\[
\frac{1}{2}\ket{11}\otimes\big(\ket{00}+\ket{01}+\ket{10}+\ket{11}\big).
\]

\paragraph{Step 4: Apply QFT\(_4\) again.}
For \(\omega = e^{i\pi/2} = i\),
\begin{align*}
\mathrm{QFT}_4\ket{00} &= \tfrac{1}{2}(\ket{00}+\ket{01}+\ket{10}+\ket{11}),\\
\mathrm{QFT}_4\ket{01} &= \tfrac{1}{2}(\ket{00}+i\ket{01}-\ket{10}-i\ket{11}),\\
\mathrm{QFT}_4\ket{10} &= \tfrac{1}{2}(\ket{00}-\ket{01}+\ket{10}-\ket{11}),\\
\mathrm{QFT}_4\ket{11} &= \tfrac{1}{2}(\ket{00}-i\ket{01}-\ket{10}+i\ket{11}).
\end{align*}

The equal superposition transforms as
\[
\mathrm{QFT}_4(\ket{00}+\ket{01}+\ket{10}+\ket{11}) = 2\ket{00}.
\]
Substituting, we find
\[
\frac{1}{2}
\big(\ket{00}-i\ket{01}-\ket{10}+i\ket{11}\big)\otimes\ket{00}.
\]

\paragraph{Step 5: Measurement outcomes.}
Measurement of the first register yields, with equal probability,
\[
\ket{00}\ket{00},\ \ket{01}\ket{00},\ \ket{10}\ket{00},\ \ket{11}\ket{00},
\]
which correspond respectively to
\[
\ket{00},\ \ket{10},\ \ket{20},\ \ket{30}.
\]

All these outcomes satisfy \(y \in S^\perp\), confirming that the measurement results are orthogonal to the hidden shift subspace, as required by Simon’s algorithm.

\section*{Discussion Summary}
\label{sec:summary}
We formulated Simon’s problem over the module $\mathbb{Z}_d^n$ and proved that, when the hidden shift $s$ has full order $\operatorname{ord}(s)=d$, a single oracle call followed by $\mathrm{QFT}_d^{\otimes n}$ produces measurement outcomes that are exactly uniform on
\[
S^\perp = \{\, y \in \mathbb{Z}_d^n : y^\top s \equiv 0 \pmod d \,\},
\]
recovering Simon’s original qubit algorithm at $d=2$.
From this uniformity we derived non-asymptotic bounds on the probability of acquiring $n-1$ independent constraints, together with explicit repetition budgets as functions of the local dimension $d$.  These bounds show that the expected number of oracle calls remains $\Theta(n)$, while the number of repetitions required to reach a target failure probability decreases with $d$.
On the implementation side, we introduced a qubit native construction of a $d$-to-one qudit oracle $f_{(d)}$ using only the original binary promise oracle $U_f$: groups of $\ell$ qubits are repacked into one effective qudit of dimension $d=2^\ell$ via a layerwise encoding. 
QuTiP simulations for $d \in \{2,3,4\}$ confirm that the sampled outcomes are confined to $S^\perp$, empirically uniform over that subspace, and that the observed repetition counts match the predicted dimension repetition tradeoff.
Taken together, these results show that qudit-style formulations of Simon’s problem and their advantages can be explored and benchmarked using standard qubit based devices and simulators, without requiring native multilevel hardware.

Future work may explore robustness under realistic noise and optimized circuit constructions for larger effective dimensions.


\printbibliography


\end{document}